\documentclass[11pt,a4paper]{article}
\usepackage{amsmath,amssymb,amsthm,amsfonts,latexsym,graphicx,subfigure}
\usepackage[affil-it]{authblk}
\usepackage[numbers,sort&compress]{natbib}
\usepackage[a4paper,margin=3cm]{geometry}
\usepackage{xcolor}

\usepackage{hyperref}

\theoremstyle{plain}
\newtheorem{theorem}{Theorem}[section]

\newtheorem{corollary}[theorem]{Corollary}

\theoremstyle{definition}
\newtheorem{definition}[theorem]{Definition}
\newtheorem{example}[theorem]{Example}
\theoremstyle{remark}
\newtheorem{remark}[theorem]{Remark}

\usepackage[affil-it]{authblk}

\begin{document}
\bibliographystyle{plain} 
\title{Estimating cellular redundancy in networks of genetic expression}
\author[1,2,3]{Raffaella Mulas}
\author[2,3]{Michael J. Casey}
\affil[1]{The Alan Turing Institute, London, UK}
\affil[2]{Mathematical Sciences, University of Southampton, UK}
\affil[3]{Institute of Life Sciences, University of Southampton, UK}
\date{}
\maketitle

\allowdisplaybreaks[4]

	\begin{abstract}
Networks of genetic expression can be modelled by hypergraphs with the additional structure that real coefficients are given to each vertex-edge incidence. The spectra, i.e.\ the multiset of the eigenvalues, of such hypergraphs, are known to encode structural information of the data. We show how these spectra can be used, in particular, in order to give an estimation of cellular redundancy, a novel measure of gene expression heterogeneity, of the network. We analyze some simulated and real data sets of gene expression for illustrating the new method proposed here.
	\end{abstract}

\section{Introduction}

Single-cell RNA-sequencing experiments are ubiquitous in cell biology, measuring the expression of each gene (number of mRNA molecules) in individual cells \cite{hwang2018single}. By measuring the whole transcriptomes of individual cells, single-cell sequencing captures the cellular heterogeneity in gene expression. Expression heterogeneity primarily arises from differential gene expression between distinct cell types and continuous variation within developing cells types \cite{trapnell2015defining, saelens2019comparison, xia2019periodic, casey2020theory}. Said molecular cell types can be reversed engineered from the measured heterogeneity through the process of unsupervised clustering \cite{kiselev2019challenges, luecken2019current, stuart2019comprehensive}.

The process of unsupervised clustering is not straightforward: most clustering methods will produce clusters from random noise, i.e.~even if no biologically relevant clusters are present in the data set, and the number of clusters is typically not known \textit{a priori} \cite{von2012clustering, vandenbon2020clustering}. But the number of clusters should correlate positively with increasing heterogeneity: the more differentially expressing cell types present in a population, the greater heterogeneity gene expression. By quantifying the total heterogeneity of a cellular population, we should be able to infer the extent (if any) of cluster structure present in the data

Measuring the total heterogeneity of the transcriptome is difficult. Various univariate, single-gene measures of heterogeneity have been developed based on variance and entropy \cite{brennecke2013accounting, love2014moderated, grun2014validation, hafemeister2019normalization, townes2019feature, sarkar2020separating, breda2021bayesian, liu2020entropy}. But cell types emerge from the concerted action of many genes: we require a multivariate measure of heterogeneity. Existing univariate measures fail to robustly generalize to the multivariate, whole-transcriptome setting due to the curse of dimensionality: variance requires some measure of distance, but distances inflate into meaninglessness when considered over many dimensions, and entropy-based measures utilize the joint distribution, which becomes increasingly poorly sampled when considered over an increasing number of genes \cite{beyer1999nearest}. 

Instead of attempting to generalize an existing univariate measure, we here derive a novel (inverse) measure of multivariate heterogeneity based on spectral hypergraph theory: an emerging field of mathematics that has been recently applied to physics \cite{MKJ,BATTISTON20201,Synchronization,Carletti,Moreno,BKJM}, and that can be applied to biochemical data analysis \cite{JostMulas2020}.

Graph theory forms a significant component of single-cell analysis (single-cell analysis referring to the analysis of single-cell data, not of single cells), particularly in dimension reduction and unsupervised clustering \cite{UMAP, PAGA, Louvain, Leiden}. By encoding data as graphs, the properties of graphs become the properties of the data, gaining us access to many powerful quantitative summaries of the data. But graphs can only partially represent single-cell data (typically encoding cell-to-cell distances): hypergraphs, a generalization of graphs, allow for a complete representation of the data. We therefore encode single-cell data as hypergraphs, and we show how recent advances in spectral hypergraph theory enable us to access the succinct structural properties of these data.

While the edges of a graph are pairwise connections between the vertices, with hypergraphs, links of any cardinality between the nodes are allowed. In other words, in a hypergraph, edges are \emph{sets} of nodes. Here we consider an even more general class of hypergraphs, namely the \emph{hypergraphs with real coefficients} that were introduced in \cite{JostMulas2020}, in which real coefficients can be assigned to the vertex-edge incidences (Figure \ref{fig:hyp}).

\begin{figure}
    \centering
    \includegraphics[width=4cm]{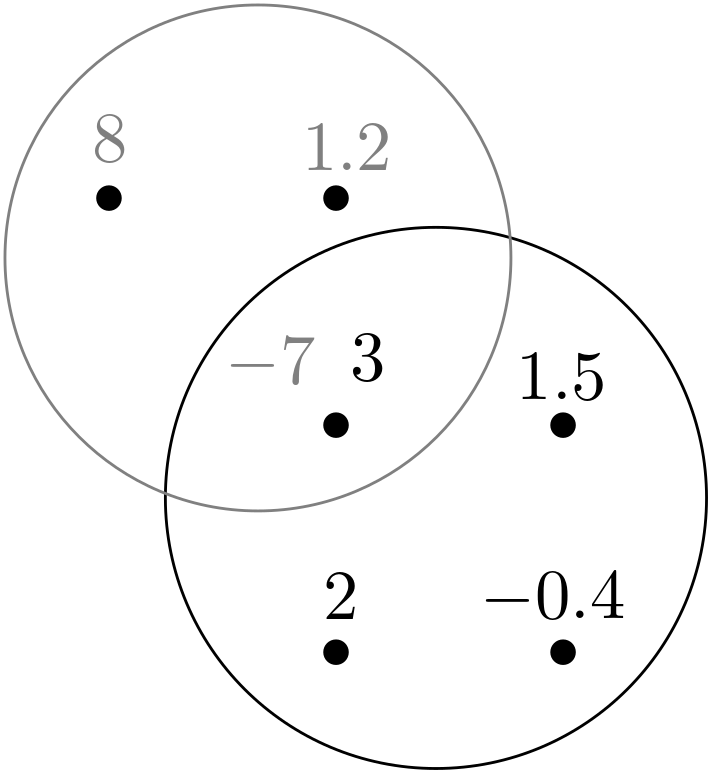}
    \caption{A hypergraph with real coefficients. In this figure, general coefficients in $\mathbb{R}$ are represented. For networks of genetic expression, we shall only consider coefficients in $[0,1]$. }
    \label{fig:hyp}
\end{figure}
Given a data set of gene expression with $N$ cells and $M$ genes, we model it as a hypergraph on $N$ nodes and $M$ edges in which each vertex $i$ represents a cell, each edge $k$ represents a gene, and each coefficient $c(i,k)$ represents the fraction of transcripts in cell $i$ mapping to gene $k$. We assume that the coefficients are normalized with respect to cells, so that 
\begin{equation}\label{eq:normalization}
\sum_k c(i,k)=1 \quad \text{for each cell }i.
\end{equation}

Such cell-wise normalization is the norm in RNA-seq analysis (e.g.~counts per million): the number of counts per cell can vary by orders of magnitude solely due to technical effects, distorting comparisons of absolute gene expression between cells \cite{dillies2013comprehensive, townes2019feature, lause2020analytic}. Accordingly, each coefficient $c(i,k)$ encodes the relative expression of gene $k$ in cell $i$.

We then consider the spectrum of the normalized Laplacian of the obtained hypergraph. As shown in \cite{JostMulas2020}, such spectrum is given by $N$ real, non-negative eigenvalues, and it encodes many important structural properties of the hypergraph, therefore also of the data. If we denote by $\lambda_N$ the largest eigenvalue, then the quantity
\begin{equation*}
\mathcal{R}:=\frac{\lambda_N}{N}
\end{equation*} estimates the \emph{cellular redundancy} of the network, an inverse measure of gene expression heterogeneity. We choose the term `redundancy' to coincide with naming of existing measures in genomics, e.g.~\cite{krakauer2002redundancy}. In fact, as we shall see in Section \ref{section_math}, we have that
\begin{equation*}
    \frac{1}{N} \leq \mathcal{R} \leq 1,
\end{equation*}and $\mathcal{R}=1/N$ if and only if and only if each gene is concentrated in one single cell, that is, we have no cellular redundancy at all. Moreover, we have that $\mathcal{R}=1$ (equivalently, $\lambda_N$ achieves its largest possible value $N$) if and only if each gene is distributed among all cells and 
\begin{equation*}
        c(i,k)=c(j,k)
\end{equation*}for each gene $k$ and for all cells $i\neq j$, i.e., we have cellular redundancy. 

We note that these two extremes of redundancy manifest at extremes of expression heterogeneity: $\mathcal{R}=1$ represents absolute homogeneity in gene expression, where every cell has the same relative expression for each gene. As $\mathcal{R}$ decreases, we move further from absolute homogeneity, arriving at a case of extreme heterogeneity when $\mathcal{R} = \frac{1}{N}$, where each gene is expressed in only a single cell. This restriction in gene expression to fewer and fewer cells as redundancy decreases corresponds to the effect of increasing differential gene expression. As the number of cell types in a population increases, the relative fraction of cells belonging to a given type, and so expressing that types marker genes, will decrease. Thus, cellular redundancy provides a multivariate measure of heterogeneity sensitive to differential gene expression between cell types.

We illustrate the meaning of this redundancy with two toy examples that involve two cells and two genes each. If genes are distributed among cells as in Table \ref{table:ex1}, we clearly expect a high cellular redundancy and in fact, by the above consideration, $\mathcal{R}=1$, hence $\mathcal{R}$ achieves its maximum possible value. Conversely, if genes are distributed among cells as in Table \ref{table:ex2}, we then have low cellular redundancy and we therefore expect a small value of $\mathcal{R}$. As we shall see in details in Example \ref{ex1} below, this is exactly the case.
    \begin{table}[h!]
    	\begin{center}
    		\begin{tabular}{ c| c | c}
    			 & Gene 1 & Gene 2  \\ 
                \hline
                Cell 1 & 0.3 & 0.7 \\
    			Cell 2 & 0.3 & 0.7 
    		\end{tabular}
    	\end{center}\caption{An example of gene expression with maximal cellular redundancy.}
    	\label{table:ex1}
    \end{table}
   \begin{table}[h!]
	\begin{center}
		\begin{tabular}{ c| c | c | c}
			& Gene 1 & Gene 2 \\ 
			\hline
	  Cell 1 & 0.1 & 0.9 \\
	Cell 2 & 1 & 0 
		\end{tabular}
	\end{center}\caption{An example of gene expression with low cellular redundancy.}
	\label{table:ex2}
\end{table}

 \subsubsection*{Structure of the paper} In Section \ref{section_math} we present the mathematical details of our model and we prove, in particular, our main theoretical results. In order to illustrate the method, in Section \ref{section:simulations} we analyze simulated data and in Section \ref{section:real-data} we analyze real data sets of gene expressions. Finally, in Section \ref{section:methods} we discuss the methods and in Section \ref{section:conclusions} we draw some conclusions.

\section{Theoretical results}\label{section_math}
We recall some definitions from \cite{JostMulas2020} before proving our main theoretical results. For completeness, we recall the general definition of a hypergraph with coefficients in $\mathbb{R}$, before specializing to the case of non-negative coefficients (Theorem \ref{thm:general} below) and to the case where the coefficients satisfy \eqref{eq:normalization} (Corollary \ref{cor:first} below). We therefore discuss a mathematical framework which is slightly more general than the one that we consider for networks of genetic expression. This allows us to offer a detailed and precise mathematical discussion, which can also be applied to other kinds of networks in future works.

	\begin{definition}
			A \emph{hypergraph with real coefficients} is a triple $\Gamma=(V,E,\mathcal{C})$ such that:
			\begin{itemize}
				\item $V=\{1,\ldots,N\}$ is a finite set of \emph{nodes} or \emph{vertices};
				\item $E=\{k_1,\ldots,k_M\}$ is a multiset of elements $k_l \in \mathcal{P}(V)\setminus \emptyset$ called \emph{edges}, where $\mathcal{P}(V)$ denotes the power set of $V$;
				\item $\mathcal{C}=\{c(i,k): i\in V\text{ and }k\in E \}$ is a set of \emph{coefficients} $c(i,k)\in\mathbb{R}$ and it is such that
				\begin{equation}\label{eq:zerocoeff}
				    c(i,k)=0 \iff i\notin k.
				\end{equation}
			\end{itemize}
		\end{definition}
		From here on in this section we fix a hypergraph $\Gamma=(V,E,\mathcal{C})$ on $N$ nodes and $M$ edges. We assume that $\Gamma$ has no isolated vertices, i.e., vertices that are not contained in any edge.
		
		
		\begin{remark}
		The fact that the edges of a hypergraph form a multiset allows us to consider distinct edges that contain the same vertices. In the case of networks of genetic expression, in particular, this allows the modelling of distinct genes which are distributed among the same cells, as for instance in the case of Table \ref{table:ex1}.
		\end{remark}
		
		\begin{definition}
    Given $k \in E$, its \emph{cardinality} is the number of vertices that are contained in $k$. Given $i\in V$, its \emph{degree} is
	\begin{equation}\label{eq:defdegree}
	\deg i:=\sum_{k\in E}c(i,k)^2.
	\end{equation}
The $N\times N$ diagonal \emph{degree matrix} of $\Gamma$ is
\begin{equation*}
    D:=\textrm{diag}\bigl(\deg i\bigr)_{i=1,\ldots,N}.
\end{equation*}
The $N\times N$ \emph{adjacency matrix} of $\Gamma$ is $A:=(A_{ij})_{ij}$, where $A_{ii}:=0$ for all $i=1,\ldots,N$ and
    \begin{equation*}
        A_{ij}:=-\sum_{k\in E}c(i,k)\cdot c(j,k)\quad \text{for all }i\neq j.
    \end{equation*}
The $N\times M$ \emph{incidence matrix} of $\Gamma$ is $\mathcal{I}:=(\mathcal{I}_{il})_{il}$, where
	\begin{equation*} 
	\mathcal{I}_{il}:=c(i,k_l).
	\end{equation*}
\end{definition}
Note that in our context, in both models that we are considering, each row of $\mathcal{I}$ represents a cell and each column of $\mathcal{I}$ represents a gene.
\begin{definition}The \emph{normalized Laplacian} of $\Gamma$ is the $N\times N$ matrix
\begin{equation*}
    L:=\mathrm{Id}-D^{-1}A= D^{-1}\mathcal{I}\mathcal{I}^\top.
\end{equation*}The \emph{dual normalized Laplacian} of $\Gamma$ is the $M\times M$ matrix
\begin{equation*}
    L^*:=\mathcal{I}^\top D^{-1}\mathcal{I}.
\end{equation*}
\end{definition}

As shown in \cite{JostMulas2020}, $L$ has $N$ real, non-negative eigenvalues, denoted
$\lambda_1\leq \lambda_2\leq \ldots\leq \lambda_N$ and called the \emph{spectrum of $\Gamma$}, while $L^*$ has $M$ real, non-negative eigenvalues, counted with multiplicity. The non-zero eigenvalues of $L$ and $L^*$ coincide, with the same multiplicity, therefore, in particular, also the largest eigenvalue of $L$ and $L^*$ is the same. This allows us to simplify the computations depending on the size of the data set: If $M\geq N$ (respectively, $N\geq M$), it is convenient to compute the largest eigenvalues and, more generally, the entire non-zero spectrum of $\Gamma$, using $L$ (respectively, $L^*$).\newline

We now prove a general result on the largest eigenvalue of $\Gamma$, before focusing on the case when the coefficients are normalized as in \eqref{eq:normalization}.
\begin{theorem}\label{thm:general}
The largest eigenvalue of $\Gamma$ is such that $1\leq \lambda_N\leq N$. Moreover, if all coefficients are non-negative,
\begin{itemize}
    \item $\lambda_N=1$ if and only if each edge has cardinality $1$;
    \item $\lambda_N=N$ if and only if all edges have cardinality $N$ and, for all vertices $i\neq j$ and edges $k\neq h$,
    \begin{equation}\label{eq:general}
        \frac{c(i,k)}{c(j,k)}=\frac{c(i,h)}{c(j,h)}.
    \end{equation}
\end{itemize}
\end{theorem}
\begin{proof}
As shown in \cite[Corollary 3.2]{JostMulas2020}, $\sum_{i=1}^N\lambda_i=N$. Since  $\lambda_N$ is the largest eigenvalue, this implies that $1\leq \lambda_N\leq N$, and $\lambda_N=1$ if and only if $\lambda_i=1$ for all $i=1,\ldots,N$. Now, since $L=\mathrm{Id}-D^{-1}A$, $1$ is the only eigenvalue of $\Gamma$ if and only if $A=0$. Since all coefficients are non-negative, this happens if and only if two distinct vertices are not contained in common edges, which can be reformulated by saying that each edge has cardinality $1$. This proves the first claim.\newline
Now, by Theorem 6.3 in \cite{JostMulas2020}, since in our case the coefficients are non-negative, $\lambda_N=N$ if and only if all edges have cardinality $N$ and there exists a function $f:V\rightarrow\mathbb{R}$ such that 
\begin{equation*}
    g(k):=c(i,k)f(i)
\end{equation*}does not depend on $i$, for all edges $k$ and for all vertices $i$. This is equivalent to saying that
\begin{equation*}
    c(i,k)f(i)=c(j,k)f(j)
\end{equation*}for all $i\neq j$, or equivalently
\begin{equation*}
    f(j)=\frac{c(i,k)}{c(j,k)}f(i).
\end{equation*}
Such $f$ exists if and only if, for all $i\neq j$, $c(i,k)/c(j,k)$ does not depend on $k$, that is,
\begin{equation*}
    \frac{c(i,k)}{c(j,k)}=\frac{c(i,h)}{c(j,h)}
\end{equation*}for all $k\neq h$ and for all $i\neq j$.
\end{proof}
 
 We now consider the case when the coefficients satisfy the normalization in \eqref{eq:normalization}.
     
     \begin{corollary}\label{cor:first}
If the coefficients satisfy \eqref{eq:normalization}, then the largest eigenvalue is such that $1\leq \lambda_N\leq N$ and, moreover,
\begin{itemize}
    \item $\lambda_N=1$ if and only if each edge has cardinality $1$;
    \item $\lambda_N=N$ if and only if all edges have cardinality $N$ and, for each edge $k$ and for all vertices $i\neq j$,
    \begin{equation}\label{eq:first}
        c(i,k)=c(j,k).
    \end{equation}
\end{itemize}
\end{corollary}

\begin{proof}By Theorem \ref{thm:general} it is enough to prove that, if the coefficients satisfy \eqref{eq:normalization}, then \eqref{eq:general} is equivalent to \eqref{eq:first}. Clearly, \eqref{eq:first} implies \eqref{eq:general}. Vice versa, \eqref{eq:general} together with \eqref{eq:normalization} implies that
\begin{equation*}
1=\sum_h c(i,h)=\Biggl(\sum_h c(j,h) \Biggr)\cdot \frac{c(i,k)}{c(j,k)}=\frac{c(i,k)}{c(j,k)}.
\end{equation*}Hence, $c(i,k)=c(j,k)$ for each edge $k$ and for all vertices $i\neq j$.
\end{proof}

In our context of networks of genetic expression, we can interpret Corollary \ref{cor:first} as follows. The largest eigenvalue $\lambda_N$ achieves its smallest possible value $1$ if and only if each gene is concentrated in one single cell. Moreover, if the coefficients satisfy the normalization in \eqref{eq:normalization}, then $\lambda_N=N$ if and only if each gene is uniformly distributed among all cells or, equivalently, we have cellular redundancy. In general, the bigger $\lambda_N$ is, the more cellular redundancy there is. We choose to consider the normalized quantity $\mathcal{R}=\lambda_N/N$ as a measure of redundancy, so that is independent of the number of cells.
 
\begin{remark}
Here we focus on the geometric and biological meaning of the largest eigenvalue for networks of genetic expression. In future works, it will be interesting to make a systematic study of the other eigenvalues as well, and to think about alternative networks that can be modelled and analyzed with these tools. As discussed in \cite{JostMulas2020}, hypergraphs with real coefficients offer a valid model for chemical reaction networks and metabolic networks. In particular, in the setting of metabolic pathway analysis, the eigenvalue $0$ is related to the \emph{metabolite balancing equation}. Also, as shown in \cite{mulas2021cheeger}, depending on the hypergraph structure, some eigenvalues may be related to clustering problems. While we do not know a priori what the meaning of the other eigenvalues in our setting is, these results suggest that a thorough study could reveal more interesting insight. In particular, we expect the other eigenvalues to have different geometric interpretations that might be translated into biological terms.
\end{remark}

We conclude this section by discussing the example from the Introduction.

\begin{example}\label{ex1} 
In the example shown by Table \eqref{table:ex2},
\begin{equation*}
\mathcal{I}=\begin{pmatrix}
0.1 & 0.9 \\ 1 & 0
\end{pmatrix} \quad \text{and} \quad D = \begin{pmatrix}
0.82 & 0 \\ 0 & 1
\end{pmatrix}.
\end{equation*}	
	The normalized Laplacian $L=D^{-1}\mathcal{I}\mathcal{I}^\top$ has eigenvalues
	\begin{equation*}
	\lambda_1\sim 0.89 \quad \text{and}\quad \lambda_2\sim 1.11.
	\end{equation*}Hence, while for a general hypergraph on $2$ nodes the redundancy $\mathcal{R}=\lambda_2/2$ is such that $0.5\leq \mathcal{R}\leq 1$, in this case we have $\mathcal{R}\sim 0.55$, i.e., there is a very low cellular redundancy, as expected.
\end{example}
 
\section{Simulations}\label{section:simulations}

We first demonstrate our hypergraph approach on simulated data. Single-cell RNA-sequencing is a counting process, measuring the number of transcripts expressed by each gene in each cell; accordingly, we simulate each gene as a Poisson process \cite{sarkar2020separating}. In practice, single-cell sequencing counts are over-dispersed compared with the Poisson, with additional variability arising from biological heterogeneity. However, for the simulation, we assume there is no excess biological variability, so parameterize each gene by a single mean parameter, $\mu$ \citep{sarkar2020separating, townes2019feature, svensson2020droplet}. We randomly generate $\mu$ for each gene by sampling from a gamma distribution (rate = 0.3; shape = 0.6) \cite{zappia2017splatter}.

We repeatedly simulate 1800 genes over a population of 300 cells (following the median ratio of genes to cells found in the real data sets shown later). Any differences between cells arise solely from stochastic variation in sampling, so we expect cells to be highly redundant: over the course of 100 simulations, we find a mean cellular redundancy of 0.842 (see Table \ref{table:sim}).

\begin{table}[h!]
\begin{center}
\begin{tabular}{ c| c c }
 Data & $\mathcal{R}$ \\ 
 \hline
 Simulated ($C=1$) & 0.842 \\
 Simulated ($C=2$) & 0.447  \\
 Simulated ($C=3$) & 0.316
\end{tabular}
\end{center}\caption{Cellular redundancy of simulations.}
\label{table:sim}
\end{table}

We recreate the effect of differentially expressing cell types in silico by splitting cells into $C$ discrete subpopulations. Each subpopulation is defined by a tranche of highly-expressed genes (of size 1800/$C$, rate = 0.3; shape = 0.6). The highly-expressing tranches of genes do not overlap between differing subpopulations, and genes not included in the high-expression tranche have on average 10x less expression (of size 1800- 1800/$C$, rate = 3; shape = 0.6).

Visualizing $C=3$ via principal component analysis (projection of high dimensional data set onto two dimensions, maximizing the variance retained), we observe that the subpopulations are clearly defined, with no overlap; moreover, the heterogeneity within each subpopulation is far less than the heterogeneity between subpopulations \cite{hotelling1933analysis}. We, therefore, expect a sharp decline in cellular redundancy with the increasing number of differentially expressing subpopulations.

Simulating each $C \in \{1,2,3\}$ 100 times, we find a strong, significant negative correlation between the number of clusters and cellular redundancy (Spearman’s rho = -0.94, p-value $<$ 2.2e-16, two-sided, N = 100): for the simulated data, cellular redundancy captures the increase in population heterogeneity.

\begin{figure}[h!]
\centering
\includegraphics[width=0.3\textwidth]{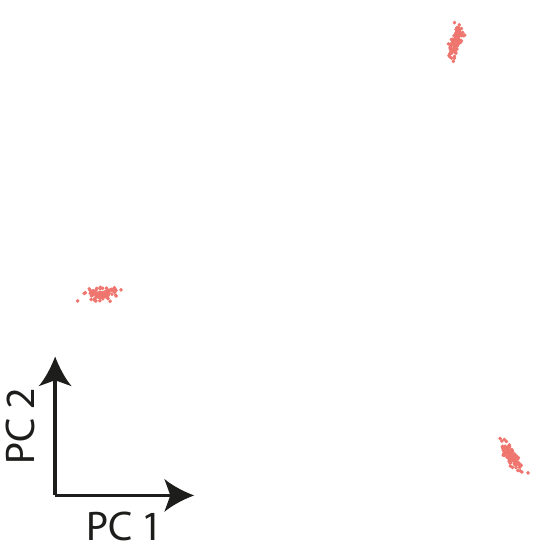}
\caption{Simulation of three distinct clusters. Cells are projected onto the first two dimensions of a principal components analysis of genes.}
\label{fig:Sim}
\end{figure}

\section{Analysis of real data sets}\label{section:real-data}
We further demonstrate our hypergraph approach on four real data sets: 1) Svensson, a technical control data set (4,000 cells-equivalents, 24,116 gene-equivalents); 2) Tian3, a mixture of three human lung adenocarcinoma cell lines, representing a simple, discrete cluster structure (902 cells, 16,468 genes), 3) Tian5, an extension of the Tian3 data set to five adenocarcinoma cell lines (3918 cells, 11,786 genes), and 4) Stumpf, a sampling from mouse bone marrow, representing a complex, continuous cluster structure (5,504 cells, 16,519 genes) \cite{svensson2017power, tian2019benchmarking, stumpf2020transfer}.

\begin{table}[h!]
\begin{center}
\begin{tabular}{ c| c c }
 Data & $\mathcal{R}$ \\ 
 \hline
 Svensson & 0.953 \\  
 Tian3 & 0.825 \\
 Tian5 & 0.771 \\
 Stumpf & 0.398
\end{tabular}
\end{center}\caption{Cellular redundancy of technical (Svensson) and biological data.}
\label{table:real}
\end{table}

\begin{figure}[b!]
\centering
\vspace{1cm}
\includegraphics[width=0.8\textwidth]{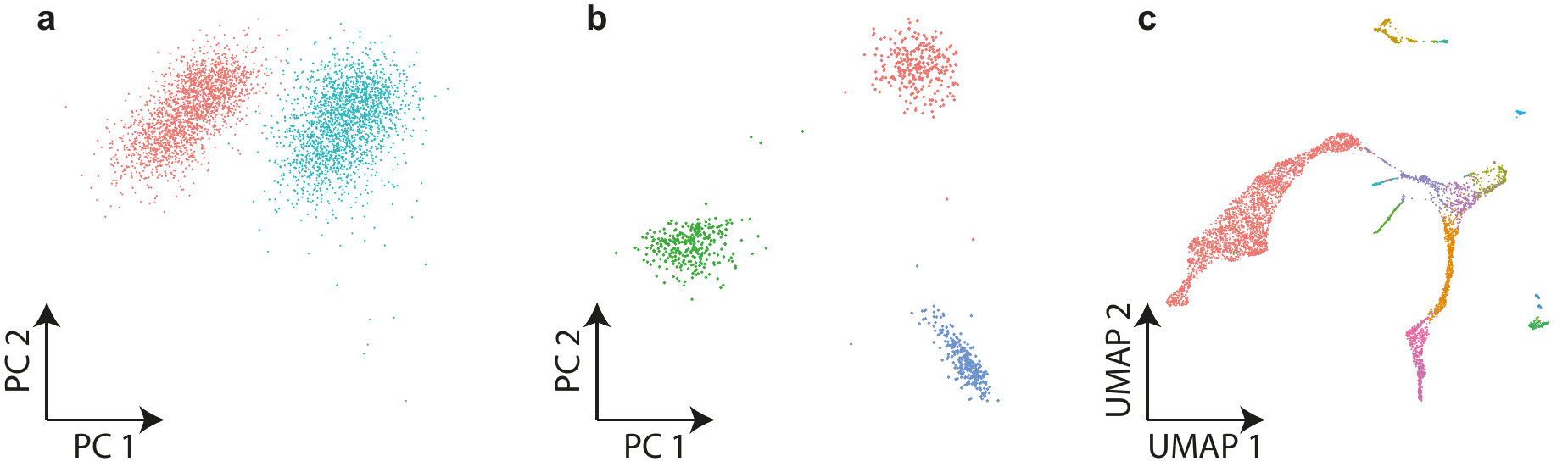}
\caption{Simulation of three data sets: \textbf{a)} Svensson, \textbf{b)} Tian3, \textbf{c} Tian5, and \textbf{d)} Stumpf. In each, cells are colored by cell type (recovered from data metadata) except for d) which is colored by experimental batch. a) - c) project cells onto principal components, and d) onto UMAP dimensions (Uniform Manifold Approximation and Projection), a nonlinear dimension reduction technique \cite{UMAP}.}
\label{fig:Real}
\end{figure}

The technical control data set was generated from a mixed solution of RNA (specifically endogenous RNA from human brain and External RNA Control Consortium spike-ins): variation between the cell-equivalents (the `cells' are not strictly cells, rather technical equivalents) arises solely from technical measurement effects \cite{svensson2020droplet}. The data consists of two experimental batches, each of which forms a distinct cluster of cells in the first two principal components, albeit with some overlap (Figure \ref{fig:Real}a). 

The Tian data sets represent a rare case of a ground truth in single-cell sequencing: as the cells come from multiple cancerous cell lines, the cluster structure is exactly specified by genotype. Tian3 and Tian5 separate into three and five clear clusters, respectively, in the first two principal components (Figure \ref{fig:Real}b\&c). Unlike with the simulated data, not all genes will be involved in the demarcation of cellular groupings\,---\,there will be substantially more overlap in gene expression between cells and so greater cellular redundancy. Moreover, even among those genes that define clusters, they may define multiple clusters or not be as differentially expressed (10-fold difference) as in the simulated data: we expect real data to have less heterogeneity from differential expression, and so display higher levels of cellular redundancy than the corresponding simulations.

This overlap in gene expression is emphasized in the Stumpf data set: cell types clusters are not discrete; instead, they form part of continuous cell differentiation trajectories (Figure \ref{fig:Real}d). Thus, while Stumpf should have the lowest cellular redundancy due to having substantially more cell types than any other data set (14), we expect the reduction in redundancy to be muted by the continuous nature of the cellular identities.

We find that both Svensson and Tian have high levels of cellular redundancy. Tian3 in particular has a cellular redundancy on par with our simulation of a single subpopulation, as the overlap in gene expression substantially increases cellular redundancy. Tian5 displays a modest decrease in cellular redundancy compared to Tian 3 (difference of 0.054), confirming the association between cellular redundancy and cell type number.

\begin{figure}[t!]
\centering
\vspace{1cm}
\includegraphics[width=1\textwidth]{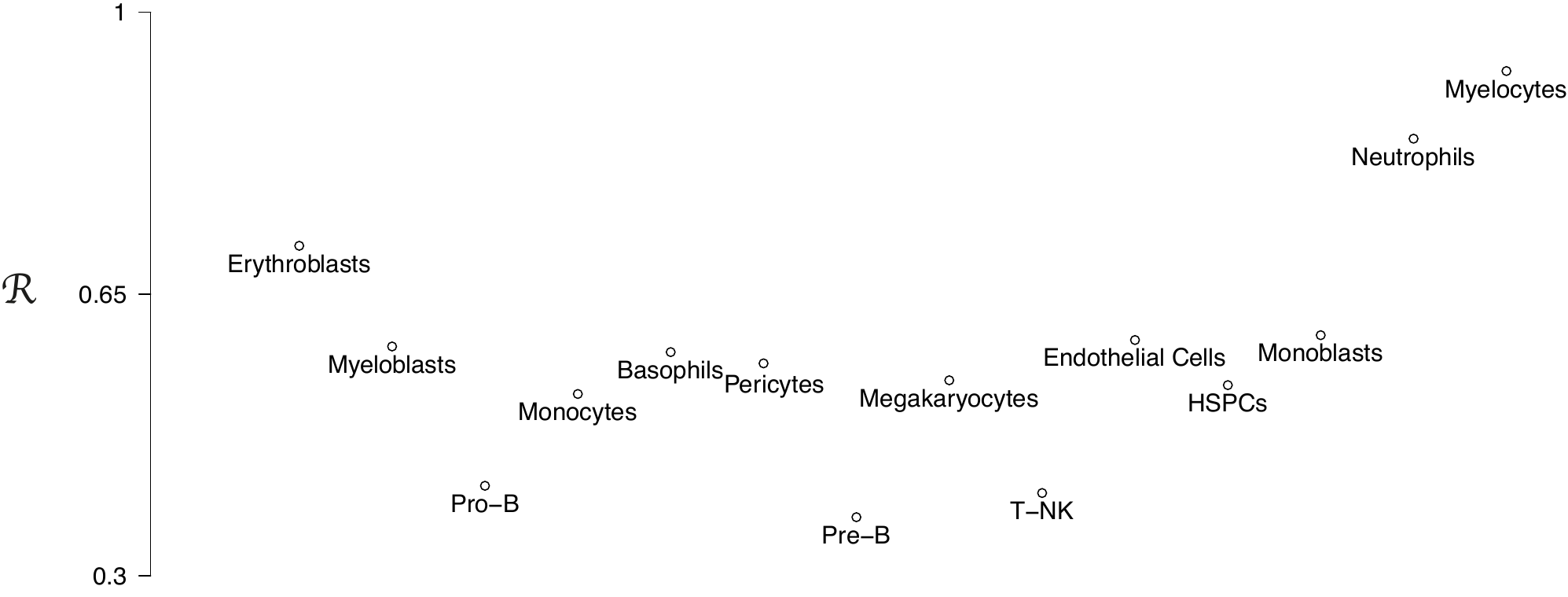}
\caption{Cellular redundancy by cell type for Stumpf data set.}
\label{fig:CellType}
\end{figure}

Stumpf has a cellular redundancy midway between the simulation of two or three subpopulations: continuous variation in gene expression substantially inflates cellular redundancy. Calculating the cellular redundancy of each cell type within the Stumpf data set, we find that most types share a similar level of redundancy, with the notable exceptions of the Pro-B, Pre-B and T-NK (Natural Killer) cell types, and the Neutrophil and Myelocyte cell types (Figure \ref{fig:CellType}). Both sets of cell types form coherent developmental lineages, the lymphocyte and neutrophil lineages respectively \citep{stumpf2020transfer}. Given the consistency otherwise of cell type's cellular redundancy, we suggest that redundancy may detect over-/under-clustering of cells: the lymphocyte-lineage `cell types' may capture a broader continuum of cellular development, so have reduced redundancy; conversely, the neutrophile-lineage `cell types' may represent cell sub-types, with greater cellular redundancy.

\section{Methods}\label{section:methods}


\subsection{Simulation}

We simulated each gene as a Poisson process, with $\mu$ randomly generated from a gamma distribution. We used a pair of gamma distributions: one with a rate of 0.3 for highly-expressed genes and one with a rate of 3 for lowly-expressed genes; both parametrizations used a shape parameter of 0.6.

For each simulation, we simulated a total of 300 cells and 1800 genes. For C of 1, i.e., one subpopulation, all genes for all cells were highly expressing; for C of 2 or 3, half or a third of genes were highly expressed for half or a third of genes, respectively. 

We calculated the mean redundancies for 100 repeats of each simulation. Spearman's rho who was calculated using the R function \textit{cor.test}.

We used the R package Seurat for principle component analyses and UMAP visualizations \cite{Seurat}.
Code for simulations and calculation of redundancy measures is available on GitHub \begin{small}(\url{https://github.com/mjcasy/scHyperGraph}).\end{small}

\subsection{Technical and Biological Data collection}

The Svensson data was downloaded as file ``svensson\_chromium\_control.h5a'' from \begin{small}\url{https://data.caltech.edu/records/1264}\end{small}\newline
The Tian data was downloaded as file ``sincell\_with\_class.RDat'' from \begin{small}\url{https://github.com/LuyiTian/sc_mixology}\end{small}\newline
The Stumpf data was downloaded as file ``RData.zip'' from \begin{small}\url{https://data.mendeley.com/datasets/csvm3kpkxd/1?fbclid=IwAR3j_hvq5Zt0cdxBBM72Fqr8zXwVC6XRpkY2JtVwbcj1gzyNLMVsdoFYWgQ} \end{small}

For each data set, we used the matrix of un-normalized, integer UMI (unique molecular identifier) counts and included all genes with at least one expressed transcript in the analysis. During analysis, each counts matrix is normalized by the total transcripts per cell so that each cell sums to 1.

\section{Conclusions}\label{section:conclusions}

We have introduced hypergraphs to the analysis of networks of genetic expression. We have identified, in particular, a measure that summarizes cellular redundancy and which can be obtained as follows. Given a network of genetic expression with $N$ cells, in which the proportions of transcripts of genes assigned to each cell sum to one, we model it as a hypergraph with real coefficients in which each vertex $i$ models a cell, each edge $k$ models a gene, and each vertex-edge coefficients $c(i,k)$ is the proportion of transcripts of $k$ assigned to $i$. The assumption on the sum of the transcripts, following existing cell-wise normalization strategies \cite{dillies2013comprehensive, townes2019feature, lause2020analytic}, can then be reformulated as $\sum_k c(i,k)=1$, for each cell/vertex $i$. As shown in \cite{JostMulas2020}, the normalized Laplacian matrix associated to this hypergraph has $N$ real, non-negative eigenvalues $\lambda_1\leq \ldots\leq \lambda_N$ whose sum is $N$. These can be easily computed and encode qualitative properties of the hypergraph. The properties of the hypergraph can be translated into properties of the data and, therefore, the eigenvalues represent a signature of the data that can be computed with little computational effort.

We have shown, in particular, that the largest eigenvalue $\lambda_N$ measures how much cellular redundancy is present in the network. Since such measure depends on the number of cells $N$, we considered $\mathcal{R}=\lambda_N/N$ as a normalized measure of cellular redundancy of the network. The choice of this measure is motivated by the theoretical results that we presented and, in order to illustrate our method, we have also analyzed both simulated and real data sets of gene expression. In both types of data, we find cellular redundancy recapitulates the expected heterogeneity of the data: redundancy reduces with an increasing number of clusters in the simulated data and with increasing biological heterogeneity in the real. The data sets lacking cluster structure, the $C=1$ simulation and Svensson data set, displayed the greatest cellular redundancies, confirming the lack of heterogeneity in those data. The closely related data sets among the real, Tian3 and Tian5 displayed a modest difference in redundancy commensurate with the increased number of clusters (3 to 5). Moreover, when broken down by cell type, cellular redundancy identified two distinct developmental lineages with increased and decreased redundancy respectively, suggesting cellular redundancy is informative not only on the number of cell types present, but the nature of those types.

In conclusion, in this work we have taken a step further in the study of spectral hypergraph theory and its applications to biology. It is worth mentioning that, while there is a vast literature on spectral graph theory applied to biochemical networks \cite{symm1,symm2,symm3,bio2006,bio2007,bio2009,bio2009JJ,bio2019}, as well as a growing literature on hypergraph modelling in biology \cite{Estrada2006,KlamtHausTheis,Stadler2015,2014signaling,2019signaling}, very little has been done, so far, in the study of spectral hypergraph theory applied to biology \cite{JostMulas2020,RRB}. Therefore, for future works, we aim to build upon the work presented here practically, applying the measure proposed here to a greater diversity of biological data, and theoretically, developing further novel measures and methods for biological data analysis from the spectra of hypergraphs.


\subsection*{Acknowledgments}
We thank the handling editor and the anonymous referees for the comments and suggestions that have greatly improved the first version of this paper. We thank Hugh Warden for providing the code used in this research. We are grateful to Tobias Goris, Jürgen Jost, Ben MacArthur and Patrick Stumpf for the interesting comments and discussions.

\subsection*{Funding}
RM was supported by The Alan Turing Institute under the EPSRC grant EP/N510129/1.

\bibliography{Genes}

\end{document}